\documentclass[11pt]{article}

\usepackage{amsfonts}
\usepackage{amssymb}
\usepackage{amsmath}
\usepackage{amsthm}
\usepackage{mathrsfs}
\usepackage{yhmath}
\usepackage{subeqnarray}

\newcommand{\cM}{{\mathcal M}}
\newcommand{\RR}{\mathbb{R}}

\newcounter{compteurlist}
\newenvironment{listeri}[1]{\begin{list}%
{\mbox{}\hfill\mbox{$(\roman{compteurlist})$}}{\usecounter{compteurlist}
\setlength{\leftmargin}{0.75cm}
\setlength{\labelwidth}{1cm}
\setlength{\rightmargin}{0cm}}
#1
}{\end{list}}
\def\endproof{\vbox{\hrule height0.6pt\hbox{%
   \vrule height1.3ex width0.6pt\hskip0.8ex
   \vrule width0.6pt}\hrule height0.6pt
  }} 

\newcommand{\ME}{\mathscr{M}}
\newcommand{\NE}{\mathscr{N}}
\newcommand{\FE}{\mathscr{F}}
\newcommand{\SE}{\mathscr{S}}

\DeclareMathOperator{\pr}{pr}
\DeclareMathOperator{\dimdiff}{dim\ diff}
\def\EOP{\ \hfill \rule{0.5em}{0.5em} }

\newtheorem{theorem}{\sc Theorem}[section]

\newtheorem{proposition}{\sc Proposition}[section]
\newtheorem{remark}{\sc Remark}[section]
\newtheorem{definition}{\sc Definition}[section]
\newtheorem{example}{\sc Example}[section]

\title{\textbf{Diffieties and Liouvillian Systems}} \author{Abdelkader
  Chelouah\thanks{The author would like to thank Prof. Michel Fliess
    for helpful discussions and his constant support during this
    work.}}

\begin{document}
\date{}

\maketitle

{\footnotesize \centerline{\textit{Laboratoire d'Informatique (LIX)}}
\centerline{\textit{\'Ecole Polytechnque}} \centerline{\textit{91128
     Palaiseau Cedex, France}}
\centerline{\texttt{chelouah@lix.polytechnique.fr}}}

\pagestyle{empty}

\begin{abstract}
  Liouvillian systems were initially introduced in \cite{kad0} and can
  be seen as a natural extension of differential flat systems. Many
  physical non flat systems seem to be Liouvillian (cf
  \cite{Sira:99ACC,Sira:99aACC,Chelouah:03FM,Kiss:IJRR2002}). We
  present in this paper an alternative definition to this class of
  systems using the language of diffieties and infinite prolongation
  theory.
\end{abstract}

\section{Introduction}
Liouvillian systems were initially defined in the differential algebra
setting. We give here a new formulation using the language of
diffieties and infinite dimensional geometries. This mathematical
framework is well suited to study Liouvillian systems.  Recall that
one of the main property of flat systems is that the variables of the
system (state, inputs) can be directly expressed, without any
integration of differential equations, in terms of the flat output and
a finite number of its time derivative. Liouvillian systems share a
similar property. To be able to derive the trajectories of a
Liouvillian system, we also need some elementary integrations called
quadratures. This can be illustrate through the following academic
example
\begin{equation}
  \label{eq:exaca}
\begin{array}{rcl}
\dot{x}_1 &=& x_2 + x_i^2,\\
\dot{x}_2 &=& x_3,\\
\dot{x}_3 &=& u.
\end{array}
\end{equation}
It is quite easy to show that (\ref{eq:exaca}) is flat for $i=1,2$ and
not flat for $i=3$ (\cite{charlet:91siam}). However, for $i=3$, the
subsystem
\begin{equation}
  \label{eq:exacaa}
\begin{array}{rcl}
\dot{x}_2 &=& x_3,\\
\dot{x}_3 &=& u.
\end{array}
\end{equation}
is flat with a flat output $y=x_2$ and the trajectory of $x_1$ can be
obtained by mean of an elementary integration $x_1=\int y+\dot{y}^2$.

For the sake of convenience, we first recall, in sections
\ref{sec:diffieties} and \ref{sec:flat}, some facts concerning the
theory of diffieties and the Lie-B\"acklund approach to equivalence
and flatness (cf
\cite{Fliess:CRAS93,Fliess:AIHP97,Fliess:MC98,Fliess:TAC99,Zharinov:book92}).
In section \ref{sec:liouvillian_systems}, we define Liouvillian
systems using the language of diffieties. Finally, we illustrate the
class of Liouvillian through the concrete case of rolling bodies
(\cite{Chelouah:03FM}, \cite{Kiss:IJRR2002}, \cite{Kiss:These}).

\section{The language of diffieties}
\label{sec:diffieties}
Let $I$ be a countable set of cardinality $\ell$, which may be finite
or not, and $\RR^I$ the linear space of all real-valued functions
$x=(x^i)$ on $I$. The space $\RR^I$ has the natural topology of the
Euclidean space if $I$ is finite and the Fr\'echet topology otherwise.
The elements $x^i$, $i\in I$, are called coordinates. For an open set
$U\subset \RR^I$ we denote by $C^{\infty}(U)$ the space of all
real-valued functions on $U$ that depend on finitely many coordinates
and are smooth as functions of a finite number of variables.  A chart
on a set $M$ is a $3$-tuple $(U,\varphi,\RR^{I})$, where $U$ is a
subset of $M$, $\varphi$ is a bijection of $U$ onto an open subset
$\varphi(U)$. The notions of smooth charts and smooth atlases can be
defined as in the finite dimensional case. The set $M$, equipped with
an equivalence class of smooth atlases, is called a $C^{\infty}$
$\RR^{I}$-manifold. The number $\ell$ does not depend on a chart
$(U,\varphi,\RR^{I})$ and is called the dimension of the smooth
manifold $M$.

A diffiety is a pair $\ME=(M,CTM)$ where $M$ is a $C^{\infty}$
$\RR^I$-manifold and $CTM$ a finite dimensional involutive
distribution on $M$. The distribution $CTM$ is called \emph{Cartan
  distribution} and its dimension the \emph{Cartan dimension} of
$\ME$. Local smooth sections of $CTM$ are called \emph{Cartan fields}.
We are only concerned here with the case of ordinary diffieties,
\textit{i.e.}, the dimension of $CTM$ is equal to $1$. For the sake of
convenience, we use without distinction the notations $(M,CTM)$ and
$(M,\partial_M)$ to denote the ordinary diffiety $\ME$, where
$\partial_M$ is a basis vector field of $CTM$.  Let $\ME=(M,CTM)$ be a
diffiety with $\dim CTM=1$. Let $(U,\varphi,\RR^I)$ be a chart on $M$
and $\partial_M$ be a basis vector field of $CTM$ on $U$, then the
$4$-tuple $(U,\varphi,\RR^I,\partial_M)$ is called a chart on $\ME$. We
denote by $\ker\vartheta_M$ the kernel of the linear map
$\partial_M:C^{\infty}(M)\rightarrow C^{\infty}(M)$, \textit{i.e.},
\begin{equation*}
  \ker\partial_M=\{\vartheta\in C^{\infty}(M)\: /\:
  \partial_M\vartheta=0\}.
\end{equation*}
A real-valued $C^{\infty}$ function $\vartheta$ on $M$ such that
$\vartheta\in\ker\partial_M$ is called a \emph{local first integral}
on $\cM$. A local first integral is said to be trivial if it is a
constant (\cite{Fliess:MC98}).  Let $\phi:M\rightarrow N$ be a smooth
mapping.  As usual, we denote by $\phi_{*}:TM\rightarrow TN$ the
differential (or tangent) mapping of $\phi$, where $TM$ (resp. $TN$)
is the tangent bundle of $M$ (resp.  $N$), and by
$\phi^{*}:T^{*}N\rightarrow T^{*}M$ the dual differential mapping of
$\phi$, \textit{i.e.}, the dual mapping of $\phi_{*}$, where $T^{*}M$
(resp. $T^{*}N$) cotangent bundle of $M$ (resp. $N$).

A smooth mapping $\phi:M\rightarrow N$ is called a
\emph{Lie-B\"acklund morphism} of a diffiety $\ME=(M,CTM)$ into a
diffiety $\NE=(N,CTN)$, written $\phi:\ME\rightarrow\NE$, if it is
compatible with the Cartan distributions $CTM$ and $CTN$,
\textit{i.e.}, $\phi_{*}(CTM)\subset CTN$.
\begin{example}
  Denote by $\RR_{\infty}^m = \RR^m\times\RR^m\times \cdots$ the
  product of a countably infinite number of copies of $\RR^m$.
  Consider the ordinary diffiety $\FE=(F,CTF)$, where
  $F=\RR\times\RR_{\infty}^m$, and let
  $(U,\varphi,\RR\times\RR^m_{\infty},\partial_{F})$ be a chart on
  $\FE$ with local coordinates $\{t,w_i^{(\nu)}\:|\: i=1,\ldots,m\:
  ; \: \nu\geq 0\}$ and basis Cartan field
  \begin{equation*}
    \partial_{F} = \frac{\partial}{\partial t} + \sum_{i=1}^{m}\sum_{\nu\geq
      0} w_i^{(\nu+1)}\frac{\partial}{\partial w_i^{(\nu)}}.
  \end{equation*}
  The diffiety $\FE$, as above defined, is usually called
  \emph{trivial diffiety} and plays a central role in the
  Lie-B\"acklund approach of flatness. On some occasions, we will use
  the short notation
  \begin{equation*}
    \partial_F = \frac{\partial}{\partial t} + \sum_{\nu\geq
  0} w^{(\nu+1)}\frac{\partial}{\partial w^{(\nu)}}.
  \end{equation*}
  to represent the basis Cartan field $\partial_F$.\EOP
\end{example}

\section{Flat systems}
\label{sec:flat}
A diffiety $\ME$ is said to be (locally) of finite type if there
exists a (local) Lie-B\"acklund submersion $\pi:\ME\rightarrow\FE$
such that the fibers are finite dimensional.  The integer $m$ is
called the (local) differential dimension of $\ME$ (cf.
\cite{Fliess:AIHP97}).
\begin{definition}[\cite{Fliess:AIHP97,Fliess:MC98}]
  A system is a (local) Lie-B\"acklund fiber bundle
  $\sigma_M=(\ME,\RR,\lambda)$, where
  \begin{listeri}
  \item $\ME$ is a diffiety of finite type where a Cartan field
    $\partial_M$ has been chosen once for all;
  \item $\RR$ is endowed with a canonical structure of a diffiety,
    with global coordinate $t$ and Cartan field $\partial/\partial t$;
  \item $\lambda:\ME\rightarrow\RR$ is a Lie-B\"acklund submersion
    such that $\lambda_{*}(\partial_M)=\partial/\partial t$.
  \end{listeri}
\mbox{}\hfill\endproof
\end{definition}
The system $\sigma_F=(\FE,\RR,\pr)$, where $\pr$ is the natural
projection mapping $\pr:\{t,w_i^{(\nu)}\}\rightarrow t$ and $\FE$ a
trivial diffiety, is called a \emph{trivial system}.

The differential dimension of a system $\sigma_M=(\ME,\RR,\lambda)$,
denoted $\dimdiff\sigma_M$, is the differential dimension of the
associated diffiety $\ME$.
\begin{definition}[\cite{Fliess:AIHP97}]
  Two systems $\sigma_M=(\ME,\RR,\lambda)$ and
  $\sigma_N=(\NE,\RR,\delta)$ are said to be (differentially)
  equivalent, written $\sigma_M\simeq\sigma_N$, if and only if
\begin{listeri}
\item $\phi_{*}(\partial_M) = \partial_N$, where
  $\phi:\ME\rightarrow\NE$ is a Lie-B\"acklund isomorphism;
\item $\lambda=\phi^{*}\delta$.
\end{listeri}
\mbox{}\hfill\endproof
\end{definition}
A system $\sigma_M=(\ME,\RR,\lambda)$ is said to be (locally)
differentially flat, or simply flat if it is (locally) equivalent to a
trivial system. If $\{t,y_i^{(\nu)}\: | \: i=1,\ldots,m\: ; \: \nu
\geq 0\}$ are local coordinates on $\FE$ then $y=(y_1,\ldots,y_m)$ is
called a flat or linearizing output.

\section{Liouvillian Systems}
\label{sec:liouvillian_systems}
A diffiety $\SE=(S,CTS)$ is called a subdiffiety of a diffiety
$\ME=(M,CTM)$ if $S$ is a submanifold of $M$ and $CTS=TS\cap CT_SM$,
\textit{i.e.}, the natural embedding $\iota:\SE\rightarrow\ME$ is a
Lie-B\"acklund immersion\footnote{Since we consider only diffieties of
  Cartan dimension $1$, $CTS=CT_SM$ here.}. The fiber bundle $T_SM$
denotes here the restriction of the vector bundle $TM$ on $S$,
\textit{i.e.},
\begin{equation*}
  T_SM = \bigcup_{p\in S} T_pM.
\end{equation*}
The tangent mapping $\iota_{*}:TS\rightarrow TM$ is injective and the
image $\iota_{*}(TS) \subset T_SM$. If $\ME$ is of finite type, then
clearly $\SE$ is of finite type as well.
\begin{definition}
  A system $\sigma_M=(\ME,\RR,\lambda)$ is said to be a differential
  extension of $\sigma_S=(\SE,\RR,\delta)$, denoted
  $\sigma_S/\sigma_M$ or $\sigma_S\subset\sigma_M$, if and only if
\begin{listeri}
\item $\SE$ is a subdiffiety of $\ME$;
\item the restriction $\iota^{*}\lambda=\delta$, where $\iota^{*}$ is
  the dual mapping of the natural embedding $\iota:\SE\rightarrow\ME$.
\end{listeri}
\mbox{}\hfill\endproof
\end{definition}
Consider the differential extension $\sigma_M/\sigma_S$, with
$\dimdiff\sigma_M=m$.  Since $\ME$ is of finite type, there exists a
Lie-B\"acklund submersion $\pi:\ME\rightarrow\FE$ such that its fibers
are finite dimensional, say $n$. Assume now that $\sigma_M$ is not
flat and $\sigma_S$ is a flat with a flat output $y=(y_1,\ldots,y_m)$.
Define the canonical bundle morphism $\rho:TM\rightarrow TM/TS$ that
takes a vector $\zeta\in T_pM$, $p\in M$, to its equivalence class
$\zeta+T_pS$ and let $\tau:TM/TS\rightarrow M$ be the fiber bundle
whose fibers $\tau^{-1}(p)$, $p\in M$, are finite dimensional. If
$\{t,\eta_1,\ldots,\eta_s,u_i^{(\nu)}\: | \:i=1,\ldots,m\: ;\: \nu\geq
0\}$ are local coordinates on $\SE$ then the Cartan distribution of
$\SE$ is spanned by
\begin{equation*}
  \partial_S=\frac{\partial}{\partial t} +
  \sum_{j=1}^{s}F_j^1\frac{\partial}{\partial\eta_j} +
  \sum_{i=1}^{m}\sum_{\nu\geq 0} u_i^{(\nu+1)}\frac{\partial}{\partial
    u_i^{(\nu)}},
\end{equation*}
where $F_j^1$ are $C^{\infty}$ functions on $S$. Using the short
notation, $\partial_S$ can be written under the form
\begin{equation*}
  \partial_S=\frac{\partial}{\partial t} +
  F^1\frac{\partial}{\partial\eta} +
  \sum_{\nu\geq 0} u^{(\nu+1)}\frac{\partial}{\partial
    u^{(\nu)}},  
\end{equation*}
with $F^1=(F_1^1,\ldots,F_s^1)$. A local smooth section $\zeta$ of
$TM/TS$ is given by
\begin{equation*}
  \zeta = \sum_{j=1}^{d=n-s} F_j^2 \frac{\partial}{\partial\xi_j} =
  F^2\frac{\partial}{\partial\xi},
\end{equation*}
where $F_j^2$ are $C^{\infty}$ functions on $M$,
$F^2=(F_1^2,\ldots,F_d^2)$, with
\begin{equation*}
  \{t,\xi_1,\ldots,\xi_{d},\eta_1,\ldots,\eta_s,u_i^{(\nu)}\: |
  \: i=1,\ldots,m\: ;\: \nu\geq 0\}
\end{equation*}
local coordinates on $\ME$.
\begin{definition}
  Let $\sigma_M$ be a differential extension of a flat system
  $\sigma_S$ and $y$ a given flat output of $\sigma_S$. Then,
  $\sigma_S$ is called a \emph{flat subsystem} of $\sigma_M$ and the
  flat output $y$ of $\sigma_S$, a \emph{partial linearizing output}
  of $\sigma_M$. If, in addition, that flat output $y$ is such that
  $d=\dim \tau^{-1}(p)$, with $p\in M$ and $\tau:TM/TS\rightarrow M$
  the aforementioned fiber bundle, is minimal, then $d$ is called the
  \emph{defect}, $\sigma_S$ a \emph{maximal flat subsystem} and $y$ a
  \emph{maximal linearizing output} of
  $\sigma_M$.\mbox{}\hfill\endproof
\end{definition}
Consider now the classical dynamics
\begin{equation}
  \label{eq:nl}
  \dot{x} = F(x,u),\qquad (x,u)\in X\times U \subset \RR^n\times\RR^m,
\end{equation}
where $x=(x_1,\ldots,x_n)$, $u=(u_1,\ldots,u_m)$ and
$F=(F_1,\ldots,F_n)$ is a m-tuple of $C^{\infty}$ functions on
$X\times U$. To (\ref{eq:nl}) we can associate a diffiety
$\ME=(\RR\times\RR^n\times\RR^m\times\RR^m_{\infty},\partial_M)$ with
local coordinates $\{t,x_1,\ldots,x_n,u_i^{(\nu)}\: | \:
i=1,\ldots,m\, ;\, \nu \geq 0\}$ and Cartan field
\begin{equation*}
  \partial_M = \frac{\partial}{\partial t} +
  \sum_{j=1}^{n} F_j\frac{\partial}{\partial x_j} +
  \sum_{i=1}^{m}\sum_{\nu\geq 0} u_i^{(\nu+1)}\frac{\partial}{\partial
    u_i^{(\nu)}}.
\end{equation*}
A subsystem of (\ref{eq:nl}) is given by a diffiety
$\SE=(S,\partial_S)$, with local coordinates
\begin{equation*}
  \{t,\eta_1,\ldots,\eta_s,u_i^{(\nu)}\: | \: i=1,\ldots,m\: ;\:
  \nu\geq 0\}
\end{equation*}
and a basis Cartan field
\begin{equation*}
  \partial_S = \frac{\partial}{\partial t} +
  \sum_{j=1}^{s}F_j^1(\eta,u)\frac{\partial}{\partial\eta_j} +
  \sum_{i=1}^{m}\sum_{\nu\geq 0} u_i^{(\nu+1)}\frac{\partial}{\partial
    u_i^{(\nu)}},
\end{equation*}
where $\eta=(\eta_1,\ldots,\eta_s)\in X^1\subset \RR^s$ and $F_j^1$
are $C^{\infty}$ functions on $X^1\times U$. A local section $\zeta$ of
$TM/TS$ is given by
\begin{equation*}
  \zeta = \sum_{j=1}^{d=n-s} F_j^2(\eta,\xi,u)
  \frac{\partial}{\partial\xi_j} =
  F^2(x,u)\frac{\partial}{\partial\xi},
\end{equation*}
where $\xi=(\xi_1,\ldots,\xi_{d})\in X^2\subset \RR^{d}$,
$F^2=(F_1^2,\ldots,F_d^2)$ and $F_j^2$ are $C^{\infty}$ functions on
$X^1\times X^2\times U=X\times U$. The vector $\xi$ represents only
the complement of $\eta$ (by renumbering the $x_i$'s if needed) to
form the vector $x$, \textit{i.e.}, $x=(\eta,\xi)$. We can assume,
in the sequel, that coordinates $\eta$ and $\xi$ are given by the
projection mappings $\pr_1$ and $\pr_2$
\begin{equation*}
  \begin{array}{rclcl}
    \pr_1 &:& \{t,x_1,\ldots,x_n,u_i^{(\nu)}\} & \rightarrow &
    \{t,\eta_1,\ldots,\eta_s,u_i^{(\nu)}\} \\
    \pr_2 &:& \{t,x_1,\ldots,x_n,u_i^{(\nu)}\} & \rightarrow &
    \{t,\xi_1,\ldots,\xi_d\}
  \end{array}  
\end{equation*}
Thus, if $\sigma_M$ is a differential extension of a flat system
$\sigma_S$, then dynamics (\ref{eq:nl}) admits the following
decomposition
\begin{equation*}
  \dot{x} = \begin{pmatrix}\dot{\eta}\\ \hdots \\ 
    \dot{\xi}\end{pmatrix} = \begin{pmatrix} F^1(\eta,u) \\ \hdotsfor{1}
    \\ F^2(\eta,\xi,u) \end{pmatrix}.
\end{equation*}
\begin{definition}
  \label{def:pv}
  A system $\sigma_M$ is called a Piccard-Vessiot extension of a system
  $\sigma_S$ if~:
  \begin{listeri}
    \item $\sigma_S$ is flat;
    \item the Cartan field $\partial_M$ is of the form
      \begin{equation*}
        \partial_M =
        A(\eta,u)\xi\frac{\partial}{\partial\xi} + \partial_S.
      \end{equation*}
    \item $\ker\partial_M=\ker\partial_S$.
    \end{listeri}
    A differential extension $\sigma_M/\sigma_S$ such that $\sigma_M$
    is a Piccard-Vessiot extension of $\sigma_S$ is said to be a
    Piccard-Vessiot extension.\mbox{}\hfill\endproof
\end{definition}
In the case of a Piccard-Vessiot system, (\ref{eq:nl}) takes the form
\begin{equation}
  \label{eq:PVSystem}
  \dot x =  \begin{pmatrix}\dot{\eta}\\ \hdots \\ 
  \dot{\xi}\end{pmatrix} = \begin{pmatrix} F^1(\eta,u) \\ \hdotsfor{1}
  \\ A(\eta,u)\xi \end{pmatrix},
\end{equation}
where $A(\eta,u)$ is a $n\times n$ matrix of $C^{\infty}(X^1\times
U)$ functions.
\begin{proposition}
  Let $\sigma_M/\sigma_S$ be a Piccard-Vessiot extension. Then,
  $\sigma_M$ is locally controllable.\hfill\endproof
\end{proposition}
\begin{proof}
  Since $\sigma_S$ is flat, $\sigma_S\simeq\sigma_F$, hence
  $\ker\partial_M=\ker\partial_S=\RR$, \textit{i.e.}, any local
  first integral of $\sigma_M$ is trivial, and it follows that
  $\sigma_M$ is locally controllable (cf \cite{Fliess:97SCL}).
\end{proof}
\begin{definition}\label{df:liou}
  Let $\sigma_M$ be a differential extension of a flat system
  $\sigma_S$ and $y$ a flat output of $\sigma_S$. Then, $\sigma_M$ is
  said to be a Liouvillian extension of $\sigma_S$, or simply
  $\sigma_M/\sigma_S$ is a Liouvillian extension, if and only if there exists a
  nested chain of subsystems $\sigma_S=\sigma_{S_0} \subset
  \sigma_{S_1} \subset \cdots \subset \sigma_{S_{d}} = \sigma_M$, with
  $\sigma_{S_j}=(\SE_j,\RR,\delta_j)$ and
  $\SE_j=(S_j,\partial_{S_j})$, such that, for $j=1,\ldots,d$,
  $\ker\partial_{S_j}=\ker\partial_{S_{j-1}}$, where either
  \begin{listeri}
    \item $\partial_{S_j} = \alpha_j
      \partial/\partial\xi_j + \partial_{S_{j-1}}$, $\alpha_j\in
      C^{\infty}(S_{j-1})$, or
    \item $ \partial_{S_j} = \alpha_j \xi_j
      \partial/\partial\xi_j + \partial_{S_{j-1}}$, $\alpha_j\in
      C^{\infty}(S_{j-1})$.
  \end{listeri}
  If $\sigma_S$ is maximal (resp. partial), \textit{i.e.}, $d$ is
  the defect of $\sigma_M$, then $\sigma_M$ is called
  \emph{Liouvillian system} (resp. partial Liouvillian system) and $y$
  \emph{Liouvillian output} (resp. partial Liouvillian output).
  \mbox{}\hfill\endproof
\end{definition}
\begin{remark}
  According to the definition, a local section $\zeta_j$ of
  $TS_j/TS_{j-1}$ is given either by
\begin{listeri}
\item $\zeta_j =
  \alpha_j(\eta,\xi_1,\ldots,\xi_{j-1},u)\partial/\partial\xi_j$ (hence
  $\dot{\xi}_j = \alpha_j(\eta,\xi_1,\ldots,\xi_{j-1},u)$ and
  $\xi_j=\int \alpha_j$) or
\item $\zeta_j =
  \alpha_j(\eta,\xi_1,\ldots,\xi_{j-1},u)\xi_j\partial/\partial\xi_j$
  (hence $\dot{\xi}_j = \alpha_j(\eta,\xi_1,\ldots,\xi_{j-1},u)\xi_j$
  and $\xi_j = e^{\int \alpha_j}$).
\end{listeri}
Hence, Liouvillian extensions are extensions by integrals (i) or
exponential of integral (ii), usually called extensions by
quadratures.\EOP
\end{remark}
Actually, it is easy to see that Liouvillian extensions are a
particular case of Piccard-Vessiot extensions. For (ii) of
definition \ref{df:liou}, it is clear that $\sigma_{S_j}$ is a
Piccard-Vessiot extension of $\sigma_{S_{j-1}}$. The extension by
integral can be obtained by considering the Piccard-Vessiot extension
$\sigma_{S_j}$ of $\sigma_{S_{j-1}}$ with Cartan field given by
\begin{equation*}
  \partial_{S_j} = \xi_{j+1}\frac{\partial}{\partial\xi_j} +
  (\dot{\alpha}_j/\alpha_j)\xi_{j+1}\frac{\partial}{\partial\xi_{j+1}},
\end{equation*}
with $\dot{\alpha}_j=d\alpha_j/dt$ and $\alpha_j\in C^{\infty}(S_j)$.
\begin{remark}
  Notice that an arbitrary linearizing output $y$ for $\sigma_S$ does not
  necessarily give rise to a Liouvillian system. Therefore, the
  Liouvillian character of a system strongly depends on the
  choice of $y$.\hfill\EOP
\end{remark}
Let $T_n(C^{\infty}(X^1\times\RR^m))$ the set of $n\times n$ lower
triangular matrices with components in $C^{\infty}(X^1\times\RR^m)$
and $U_n(C^{\infty}(X^1\times\RR^m))$ the subset of
$T_n(C^{\infty}(X^1\times\RR^m))$ such that all the diagonal
components are equals to $1$.
\begin{theorem}
  Consider the Piccard-Vessiot system (\ref{eq:PVSystem}). If
  $A(\eta,u)\in U_d(C^{\infty}(X^1\times\RR^m))$ then
  (\ref{eq:PVSystem}) is Liouvillian.\hfill\endproof
\end{theorem}
\begin{proof}
  Let $A=(a_{i,j})_{i,j=1,d}\in U_d(C^{\infty}(X^1\times\RR^m))$,
  then
  \begin{equation*}
    \dot\xi_i=\sum_{j=1}^{i}a_{i,j}\xi_j,\quad \textrm{with } a_{i,i}=1,
    \textrm{ }i=1,\ldots,d.
  \end{equation*}
  In particular,
  \begin{equation*}
    \dot\xi_1 = \xi_1,
  \end{equation*}
  and it follows that $\xi_1$ is an exponential of integral. Next,
  \begin{equation*}
    \dot{\wideparen{\left(\frac{\xi_2}{\xi_1}\right)}} =
    \frac{\dot\xi_2}{\xi_1}-\frac{\xi_2}{\xi_1} = a_{2,1},
  \end{equation*}
  in the other words,
  \begin{equation*}
    \xi_2 = \xi_1\int a_{2,1}.
  \end{equation*}
  Finally, differentiating $\xi_i/\xi_1$, $i=1,\ldots,d$, gives
  \begin{equation*}
    \dot{\wideparen{\left(\frac{\xi_i}{\xi_1}\right)}} = a_{i,1} +
    a_{i,2}\frac{\xi_2}{\xi_1} + \ldots +
    a_{i,i-1}\frac{\xi_{i-1}}{\xi_1}.
  \end{equation*}
  Now, making the appropriate induction assumption, we deduce that
  $\xi_i$, $i=2,\ldots,d$, can be obtained by mean of the integral
  \begin{equation*}
    \xi_i=\xi_1\int a_{i,1} +
    a_{i,2}\frac{\xi_2}{\xi_1} + \ldots +
    a_{i,i-1}\frac{\xi_{i-1}}{\xi_1},
  \end{equation*}
  which concludes the proof.
\end{proof}
\begin{theorem}
  Consider the Piccard-Vessiot system (\ref{eq:PVSystem}). If
  $A(\eta,u)\in T_d(C^{\infty}(X^1\times\RR^m))$ then
  (\ref{eq:PVSystem}) is Liouvillian.\hfill\endproof
\end{theorem}
\begin{proof}
  Let $A=(a_{i,j})_{i,j=1,d}\in T_d(C^{\infty}(X^1\times\RR^m))$,
  \textit{i.e.},
  \begin{equation*}
    \dot\xi_i=\sum_{j=1}^{i}a_{i,j}\xi_j,\quad 
    \textrm{ }i=1,\ldots,d.
  \end{equation*}
  and set
  \begin{equation*}
    \dot c_i = a_{i,i},\quad i=1,\ldots,d.
  \end{equation*}
  So the $c_i$'s are integrals of elements of the flat subsystem
  $\sigma_S$. For $i=1$, we get
  \begin{equation*}
    \dot\xi_1 = a_{1,1}\xi_1,
  \end{equation*}
  and it follows that $\xi_1=e^{c_1}$. Next,
  \begin{equation*}
    \dot{\wideparen{\left(\frac{\xi_2}{\xi_1}\right)}} =
    \frac{\dot\xi_2}{\xi_1} - \frac{\dot\xi_1}{\xi_1}\frac{\xi_2}{\xi_1} =
    a_{2,1} + (a_{2,2}-a_{1,1})\frac{\xi_2}{\xi_1},
  \end{equation*}
  in the other words,
  \begin{equation*}
    \xi_2 = \xi_1 e^{c_2-c_1}\int a_{2,1} e^{c_1-c_2}.
  \end{equation*}
  Finally, differentiating $\xi_i/\xi_1$, $i=2,\ldots,d$, gives
  \begin{eqnarray*}
    \dot{\wideparen{\left(\frac{\xi_i}{\xi_1}\right)}} = a_{i,1} +
    a_{i,2}\frac{\xi_2}{\xi_1} + \ldots
    + a_{i,i-1}\frac{\xi_{i-1}}{\xi_1}
    +(a_{i,i}-a_{1,1})\frac{\xi_i}{\xi_1}.
  \end{eqnarray*}
  Now, making the appropriate induction assumption, we deduce that
  $\xi_i$, $i=2,\ldots,d$, can be obtained by mean of the relation
  \begin{eqnarray*}
    \xi_i &= \xi_1 e^{c_i-c_1} \cdot
    \int \left(a_{i,1} + a_{i,2}\frac{\xi_2}{\xi_1} + \ldots +
      a_{i,i-1}\frac{\xi_{i-1}}{\xi_1}\right)e^{c_1-c_i},
  \end{eqnarray*}
  which concludes the proof.
\end{proof}
\section{The rolling bodies}
Let us illustrate the class of Liouvillian through the concrete case
of rolling bodies (see \cite{Chelouah:03FM} for a more larger
treatment). The kinematic equations of motion of the contact point
between two bodies rolling on top of each other are given in geodesic
coordinates by
\begin{equation}
  \begin{array}{rcl}
    \dot{v}_1 &=& u_1, \\
    \dot{w}_1 &=& \mbox{$\frac{1}{B}$} u_2, \\
    \dot{v}_2 &=& u_1\cos\psi - u_2\sin\psi, \\
    \dot{w}_2 &=& -\mbox{$\frac{1}{C}$} (u_1\sin\psi +
    u_2\cos\psi), \\
    \dot{\psi} &=& \mbox{$\frac{B_{v_1}}{B}$}u_2 -
    \mbox{$\frac{C_{v_2}}{C}$} (u_1\sin\psi + u_2\cos\psi),
  \end{array}
  \label{rolling-surfaces-system}
\end{equation}
where $B=B(v_1,w_1)$, $C=C(v_2,w_2)$, $B_{v_1}=\partial B/\partial
v_1$ and $C_{v_2}=\partial C/\partial v_2$. In the case of the well
known plate ball problem, $B\equiv 1$ and $C=\cos(v_2)$ and
(\ref{rolling-surfaces-system})
takes the form
\begin{equation}
  \begin{array}{rcl}
    \dot{v}_1 &=& u_1, \\
    \dot{w}_1 &=& u_2, \\
    \dot{v}_2 &=& u_1\cos\psi - u_2\sin\psi, \\
    \dot{w}_2 &=& -\frac{1}{\cos v_1} (u_1\sin\psi +
    u_2\cos\psi), \\
    \dot{\psi} &=& \tan v_2 (u_1\sin\psi + u_2\cos\psi).
  \end{array}
  \label{rolling-surfaces-system-pb}
\end{equation}
This system is Liouvillian and a Liouvillian output is given
by (see \cite{Chelouah:03FM} for the details)
\begin{equation}
  \label{eq:mlopb}
  \begin{array}{rcl}
  x &=& v_1 - v_2\cos\psi,\\
  y &=& w_1 + w_2\sin\psi.
  \end{array}
\end{equation}
The idea to give a new formulation of Liouvillian systems within the
mathematical framework of diffieties was actually motivated by the
study of the rolling bodies system. As matter of fact, it is not clear
whether this system is Liouvillian within the differential algebraic
setting (cf \cite{kad0}).

First, notice that system \ref{rolling-surfaces-system-pb} is not
under a suitable form to describe a system in the differential
algebraic setting (the associated differential field extension in not
finitely generated). However, using transformations
$\sigma=\tan(\psi/2)$ and $\xi=\tan(v_2/2)$,
(\ref{rolling-surfaces-system}) writes (plate ball case)
\begin{equation}
  \label{eq:RollingBodiesAlgebraic}
  \begin{array}{rcl}
    \dot{v}_1 &=& u_1,\\
    \dot{w}_1 &=& u_2,\\
    \dot{\xi} &=&
    \frac{1+\xi^2}{2(1+\sigma^2)}\big[(1-\sigma^2)u_1-2\sigma u_2\big],\\
    \dot{w}_2 &=& -\frac{1+\xi^2}{(1-\xi^2)(1+\sigma^2)}\big[2\sigma u_1 +
    (1-\sigma^2)u_2\big],\\
    \dot{\sigma} &=& \frac{\xi}{1-\xi^2}\big[2\sigma u_1 +
    (1-\sigma^2)u_2\big],
  \end{array}
\end{equation}
and becomes explicit and rational. The associated differential ideal
is thus prime and leads to finitely generated differential field
extension. Nevertheless, the Liouvillian output (\ref{eq:mlopb}) takes
now the form
\begin{equation}
  \label{eq:mlopbalg}
  \begin{array}{rcl}
    \tilde x &=& v_1 -\frac{1-\sigma^2}{1+\sigma^2}\: 2\arctan\xi, \\
    \tilde y &=& w_1 + \frac{2\sigma}{1+\sigma^2}\:w_2.
  \end{array}
\end{equation}
If we denote by $\RR\langle v_1,w_1,\xi,w_2,\sigma,u_1,u_2\rangle$ the
differential field generated by $\RR$ and the variables
$\{v_1,w_1,\xi,w_2,\sigma,u_1,u_2\}$, and $\overline{\RR\langle
  v_1,w_1,\xi,w_2,\sigma,u_1,u_2\rangle}$ its algebraic closure, then
$\tilde x$ and $\tilde y$ are not in $\overline{\RR\langle
  v_1,w_1,\xi,w_2,\sigma,u_1,u_2\rangle}$, and it follows that
(\ref{eq:mlopbalg}) is not a Liouvillian output for
(\ref{eq:RollingBodiesAlgebraic}) in this context.

\bibliographystyle{plain}

\end{document}